\documentclass[10pt]{article}
\usepackage{fullpage,graphicx}
\usepackage{amssymb,amsmath,amsthm}
\usepackage{natbib}
\title{Sign-constrained least squares estimation for high-dimensional regression}
\author{Nicolai Meinshausen \\ University of Oxford, UK \\ meinshausen@stats.ox.ac.uk }
\newcommand{\mb}{\mathbf} 
\renewcommand{\S}{\mb \Sigma}
\newtheorem{assumption}{Assumption}
\newtheorem{theorem}{Theorem}
\newtheorem{lemma}{Lemma}
\newtheorem{corollary}{Corollary}
\newtheorem{remark}{Remark}

\begin{document}
\maketitle
\begin{abstract} Many regularization schemes for high-dimensional regression have been put forward. 
Most require the choice of a tuning parameter, using model selection criteria or cross-validation schemes. We show that a simple non-negative or sign-constrained least squares is a very simple and effective regularization technique for a certain class of high-dimensional regression problems. The sign constraint has to be derived via prior knowledge or an initial estimator but no further tuning or cross-validation is necessary.  The success depends on  conditions that are easy to check in practice. A sufficient  condition for our results is that most variables with the same sign constraint are positively correlated.  For a sparse optimal predictor,  a non-asymptotic bound on the L1-error of the regression coefficients is then proven. Without using any further regularization, the regression vector can be estimated consistently as long as $\log(p)s/n\rightarrow 0$ for $n\rightarrow\infty$, where $s$ is the sparsity of the optimal regression vector, $p$ the number of variables and $n$  sample size. Network tomography is shown to be an application where the necessary conditions for success of non-negative least squares are naturally fulfilled and empirical results confirm the effectiveness of the sign constraint for sparse recovery.
\end{abstract}

\section{Introduction}

High-dimensional regression problems are characterized by a large number of predictor variables in relation to sample size.  Regularization (in a broad sense) is of critical importance for high-dimensional problems and much attention has been paid to various schemes and their properties in recent years, including the \emph{Ridge} estimator \citep{hoerl1970ridge},  \emph{non-negative Garrotte} \citep{breiman95better}, the \emph{Lasso} \citep{tibshirani96regression} and various variations of the latter, including the \emph{group Lasso} \citep{yuan2006model} and \emph{adaptive Lasso} \citep{zou05adaptive}.
Datasets with very low signal-to-noise ratio offer similar challenges to high-dimensional problems even if the notional sample size is quite high.

Sign-constraints on the regression coefficients are a simpler regularization and have been first advocated by I.J. Good, as covered in the book \citet{lawson1995solving}. There is a wide range of problems where the sign of the regression coefficients can either be estimated by an initial estimator or where it is known a priori, such as in image processing and spectral analysis \citep{waterman1977c1,bellavia2006interior,donoho1992maximum,chen2009nonnegativity}. Sign-constraints have also been implemented for matrix factorizations, specifically the \emph{non-negative Matrix factorization} \citep{lee1999learning,lee2001algorithms,ding2010convex} and  \emph{non-negative least squares} regression can be a useful tool for this factorization \citep{kim2007sparse}.  We study the performance of non-negative least squares type problems under a so-called \emph{Positive Eigenvalue Condition}, which can be checked for any given dataset by solving a quadratic programming problem. A sufficient condition uses only the minimum of all entries in the design matrix.  It is shown that non-negative (or, in general, sign-constrained) least squares is a surprisingly effective regularization technique for high-dimensional regression problems under these conditions.  If the \emph{Positive Eigenvalue Condition} is not fulfilled, the sign constraint is still a good ingredient in a regularization framework. The \emph{non-negative Garrote} \citep{breiman95better} is, for example, making use of a sign-constraint, where the signs are derived from an initial estimator as is the \emph{positive Lasso} \citep{efron04least}. 

The data are assumed to be given by a $n\times 1$-vector of real-valued observations $\mb Y$ and a $n\times p$-dimensional matrix $\mb X$, where  column $k$ of $\mb X$ contains all $n$ samples of the $k$-th predictor variable for $k=1,\ldots,p$.
The non-negative least squares (NNLS)  regression estimator is defined as 
\begin{equation} \label{eq:nnls}      \hat{\beta} \; :=\; \mbox{argmin}_{\beta} \; \| \mb Y - \mb X \beta\|_2^2 \quad \mbox{such that} \;\; \min_k\beta_k \ge 0  . \end{equation}
We will work with a positivity constraint without limitation of generality since variables that are constrained to be negative can be replaced by their negative counterpart and the problem can thus always be framed as a non-negative least squares optimisation. Problem (\ref{eq:nnls}) is a convex optimization problem and can be solved with general quadratic programming problem solvers, including active set \citep{lawson1995solving}, iterative \citep{kim2006new} and interior-point approaches \citep{bellavia2006interior}. A tailor-made fast approximate algorithm based on random projections has recently been proposed in \citet{boutsidis2009random}. 
The recent manuscript \citet{slawskisparse} contains independent work on the behaviour of NNLS in high-dimensions. Using the same \emph{Positive Eigenvalue Condition} (which is called self-regularizing design condition),
 a bound on the prediction error of NNLS and a sparse recovery property after hard thresholding  are shown in \citet{slawskisparse}. Our main focus is on sparse recovery in the $\ell_1$-sense. The bounds on prediction error are also of different nature since the assumptions are different. We make use of the so-called compatibility condition which is appears in most sparse recovery results in the $\ell_1$-norm penalized estimation literature \citep{van2009conditions} and derive, with the help of this condition, tight non-asymptotic bounds on the prediction error.

Note that the non-negative least squares estimator (\ref{eq:nnls}) does not require the choice of a tuning parameter beyond choosing the sign of the coefficients. Imposing a sign-constraint might seem like a very weak regularization but it will be shown that the estimator is remarkably different from the un-regularized least squares estimator. It can cope with high-dimensional problems, where the number of predictor variables vastly exceeds sample size. It will be shown to be a consistent estimator as long as the underlying optimal prediction is sufficiently sparse (ie using only a small subset of all predictor variables) and the so-called \emph{Positive Eigenvalue Condition} is fulfilled.

The manuscript is organized as follows. The notation and the main two assumptions, the compatibility and \emph{Positive Eigenvalue Condition}, are introduced in Section~\ref{sec:notation}. Our main result, a $\ell_1$-bound on the difference between the NNLS estimator and the optimal regression coefficients, is shown in Section~\ref{sec:main}, along with a bound on the prediction error.

\section{Notation and Assumptions}
\label{sec:notation}
We assume that the $n$ samples $\mb Y\in\mathbb{R}^n$ are drawn from $\mb X \beta^* +\varepsilon$ for some p-dimensional vector $\beta^*$ with $\min_k\beta^*_k\ge 0$ and $\varepsilon \sim \mathcal{N}(0,\sigma^2)$ for some $\sigma>0$.
Let $S$ be the set of non-zero entries of the optimal solutions, $S:\{k:\beta^*_k\neq 0 \}$ and $N=S^c$ be the complement of $S$. We could also let $\beta^*$ be the best approximation to the data-generating model under positivity constraints but will refrain from doing so for notational simplicity. 
We assume that the columns of $\mb X$ are standardized to $\ell_2$-norm of $n$. Despite not necessarily assuming that the columns are mean-centered, we call  $\hat{\S}= n^{-1} \mb X^T \mb X$ the covariance matrix throughout. 

We make two major assumptions for the main result, one about sparse eigenvalues and another about the positive eigenvalue between predictor variables. 

\subsection{Compatibility Condition}
There has been much recent work on the properties of the Lasso \citep{tibshirani96regression}. Many similar conditions for success of the Lasso penalization schemes have been derived \citep[for example][]{zhang06model,meinshausen06lassotype,wainwright06sth,bunea06agr,bunea06sparsity,van2008high,bickel07dantzig}. A good overview of all conditions and their relations is given in \citet{van2009conditions}. The weakest condition is based on the notion of $(L,S)$ restricted $\ell_1$-eigenvalues.

The $(L,S)$ restricted $\ell_1$-eigenvalue of matrix $\mb A$ is defined as: 
 \[ \phi^2_{\mathit{compatible}}(L,S,\mb A) := \min\Big\{ s \frac{\beta^T \mb A\beta }{ \| \beta \|_1^2 } : \beta \in \mathcal R(L,S) \Big\},\] 
where $\mathcal R(L,S) = \{\beta: \| \beta_N\|_1 \le L \| \beta_S\|_1\}$ and $s=|S|$.

A lower bound on this restricted eigenvalue is necessary for success of the Lasso, either in a prediction loss or coefficient recovery sense and was called the compatibility condition in \citet{van2009conditions}. It was shown to be weaker than all similar conditions such as the \emph{Restricted Isometry Property} \citep{candes2005dss}.

We make the following assumption.
\begin{assumption}[Compatibility Condition]
There exists some $\phi >0$ such that the $(L,S)$-restricted $\ell_1$-eigenvalue 
$\phi^2_{\mathit{compatible}}(L,S,\hat{\S}) \ge \phi.$
\end{assumption}
The value of $L$ will be specified in Theorem~\ref{theorem:main}. 

\begin{remark}
The  assumption is formulated for the empirical covariance matrix $\hat{\S}$ but can also easily be reformulated on the population covariance matrix $\S$ for random design.
Assume that the  maximal difference between the population and empirical covariance matrix is bounded by $\delta >0$, that is 
$\|\hat{\S} - \S\|_\infty \le \delta.$ This assumption is fulfilled with high probability for many data sets with larger  sample size. If the predictors have for example a multivariate normal distribution (which will not be assumed elsewhere), then the condition is fulfilled with probability $1-2\exp(-t)$ for $\delta \ge \sqrt{u} +u$ with $u=(4t+8\log(p))/n$, see (10.1) in \citet{van2009conditions}.
If $\delta\le \phi^2/ (4(L+1)^2s)$, then $\phi^2_{\mathit{compatible}}(L,S,\S) \ge \phi$ implies $\phi^2_{\mathit{compatible}}(L,S,\hat{\S}) \ge \phi/2$.
The proof follows from the inequality  $\phi^2_{\mathit{compatible}}(L,S,\hat{\S}) \ge \phi^2_{\mathit{compatible}}(L,S,\S) -(L+1) \sqrt{\delta s}$  in Corrolary~10.1 in 
\citet{van2009conditions}. 
The \emph{Compatibility Condition} could thus  be imposed on the population covariance matrix instead of the empirical covariance matrix. 
\end{remark}

\subsection{Positive eigenvalue condition}

The following \emph{Positive Correlation Condition} is the main assumption necessary to show success of non-negative least squares. 

The positively constrained minimal $\ell_1$- eigenvalue of matrix $\mb A$ is defined as 
\[ \phi^2_{pos}(\mb A) := \min\Big\{ \frac{\beta^T \mb A \beta}{ \|\beta\|_1^2} : \min_k \beta_k \ge 0\Big\},\]
A lower bound on this restricted eigenvalue will be a sufficient condition for sparse recovery success of NNLS.
\begin{assumption}[Positive Eigenvalue Condition]
There exists some $\nu > 0 $ such that $\phi^2_{pos}( \hat{\S}) \ge \nu$.
\end{assumption}

A lower bound on this eigenvalue seems to be a much stricter condition than the \emph{Compatibility Condition}. However, the latter allows for positive and negative regression coefficients, while the \emph{Positive Eigenvalue Condition} is restricted to positive coefficients. There are thus some immediate examples where it is fulfilled, which we discuss below.

\paragraph{Example I: strictly positive covariance matrix.} 
The \emph{Positive Correlation Condition} is fulfilled if $\min_{i,j} \hat{\S} \ge \nu >0$, that is all entries in the covariance matrix are strictly positive. Again, this condition could also be formulated for the population covariance matrix, using a bound on $\|\S - \hat{\S}\|_\infty$.

We also remark on the case of general sign-constraints (some variables constrained to be positive, some negative). The condition applies then to the dataset where all variables with a negativity constraint have been replaced with their negative counterparts.  The constraint on the original covariance matrix is thus that it forms two blocks. The variables in the first block are the variables with a positivity constraint and the second block is formed by all variables with a negativity constraint. Correlations are required to be positive within a block and negative between blocks. 

A generalization of Example I is the following.

\paragraph{Example II: only few negative entries.} 
Let $\mathcal{A}:=\{i : \hat{\S}_{ij}<0 \mbox{  for some } 1\le j\le p\}$ be the minimal set such that $\hat{\S}_{ij}<0$ implies $\{i,j\}\subseteq \mathcal{A}$ for all $1\le i,j\le p$. The \emph{Positive Eigenvalue Condition} is fulfilled if both of the conditions below are fulfilled  for some $\nu>0$.
\begin{enumerate}
\item All entries of the covariance matrix are strictly positive on  $\mathcal{A}^c$, that is $\hat{\S}_{ij} \ge 2\nu $ if $\{i,j\}\subseteq \mathcal{A}^c$ for all $1\le i,j\le n$.
\item A restricted eigenvalue condition holds on the set $\mathcal{A}$, ie \[ \min \Big\{ \frac{ \beta^T \hat{\S} \beta}{\|\beta\|_1^2}: \beta_k=0 \mbox{   for all  } k\in \mathcal{A}^c\Big\} > 2\nu.\]
\end{enumerate}

If the set $\mathcal{A}$ is very small, in particular much smaller than $n$, the latter restricted $\ell_1$-eigenvalue condition is in general not very restrictive. The important criterion is thus whether the set $\mathcal{A}$ is small compared to the sample size.

\paragraph{Example III: block matrix.} For a $p\times p$-matrix $\mb A$ and a set $K\subseteq\{1,\ldots,p\}$, let $\mb A_{KK}$ be the $|K|\times |K|$-submatrix formed by all elements in set $K$. Suppose 
\begin{enumerate}
\item Entries of the covariance matrix can be negative but fulfil $\hat{\S}_{ij} \ge - \rho/p^2$ for all $1\le i,j\le n$ and some $\rho>0$.
\item The set of variables $\{1,\ldots,p\}$ can be partitioned into $B\ge 1$ blocks $B_j\subseteq \{1,\ldots,p\}$ such that $\phi_{pos}^2(\hat{\S}_{B_jB_j}) \ge (\nu +\rho) B$ for all $j=1,\ldots,B$.
\end{enumerate}
A more specific example is thus: all entries in $\hat{\S}$ are larger than $-\rho/p^2$ for some $\rho>0$ and  $\hat{\S}_{ij} \ge (\nu +\rho)B$ if both $i,j$ are within the same block.

The \emph{Positive Eigenvalue Condition} is fulfilled with parameter $\nu>0$.

The positive aspect of the condition is that it is very easy to check in practice whether it applies (at least approximately) and whether one would thus expect the bounds shown below to apply to a given dataset.

\section{Main Results}
\label{sec:main}
It will be shown that non-negative least squares leads to a good recovery of the optimal sparse regression vector for high-dimensional data. We study the $\ell_1$-error in the regression vector, which also yields a bound on the $\ell_2$-error and prediction loss.

\begin{theorem}\label{theorem:main}
Assume that the \emph{Positive Eigenvalue Condition} holds with $\nu>0$. Choose any $0<\eta<1/3$. Assume that the compatibility condition holds with $\phi>0$  for $L=4\nu^{-1} $. Setting 
\[ K_{p,\eta}^2:=  2 \log\big(\frac{\sqrt 2 p}{\sqrt{\pi} \eta}\big) \]   and assuming $\min_{k\in S }\beta_k >  K_{p,\eta} \sigma/\sqrt{n \phi}$, it then holds with probability at least $1-\eta$ that  
\begin{eqnarray}
\|\hat{\beta} - \beta^*\|_1 & \le &   K_{p,\eta} \, (5/\nu +4 /\sqrt{\phi}) \,\frac{s\sigma}{\sqrt{n}} \label{eq:toshow}
\end{eqnarray}
\end{theorem}
A proof is in the appendix. 

The result might be surprising since it implies that non-negative least squares is succeeding in recovering the regression coefficients in an $\ell_1$-sense if $\log(p) s /\sqrt{n} \rightarrow 0$ for $n \rightarrow \infty$, a scaling that requires for general design a lot more regularization in the form of Lasso penalties (or similar). 

The result does not imply exact sign recovery in the sense that the non-zero coefficients equal exactly the set $S$ (and indeed this will in general not be the case), but it implies that the $S$ largest coefficients correspond to the variables in the set $S$.
\begin{corollary}\label{corr:main}
Under the same conditions as Theorem~\ref{theorem:main} and the stronger assumption  that the minimum over all non-zero coefficients is bounded from below by  $ \min_{k\in S }\beta_k \ge 2K_{p,\eta}\sigma \, (5/\nu +4 /\sqrt{\phi}) s/\sqrt{n}  ,$ it  holds with probability at least $1-\eta$ that the indices of the $s$ largest absolute coefficients in $\hat{\beta}$ are identical to the set $S$.
\end{corollary}
This follows immediately from Theorem~\ref{theorem:main} since the $\ell_1$-bound on the difference between $\hat{\beta}$ and $\beta^*$ implies the same bound in the supremum-norm.

The bound in Theorem~\ref{theorem:main} also implies a bound on the prediction error.
\begin{theorem}\label{theorem:pred}
Under the same conditions as Theorem~\ref{theorem:main}, with probability at least $1-\eta$ for any $0<\eta<1/3$,
\begin{eqnarray*}
\| \mb X (\hat{\beta}^{oracle} - \hat{\beta}) \|_2^2 \; \le \;   2 K_{p,\eta}^2\sigma^2 \,( 5/\nu + 2 /\sqrt{\phi})\, s .
\end{eqnarray*}
\end{theorem}
A proof is given in the appendix. 
The mean squared error, introduced by using NNLS instead of the oracle estimator is thus proportional to $\log(p)^2 s/n$.
The result implies asymptotically vanishing prediction error if $s\log(p)^2/n\rightarrow 0$ for $n\rightarrow n$.

\begin{figure}
\begin{center}
\includegraphics[width=0.49\textwidth]{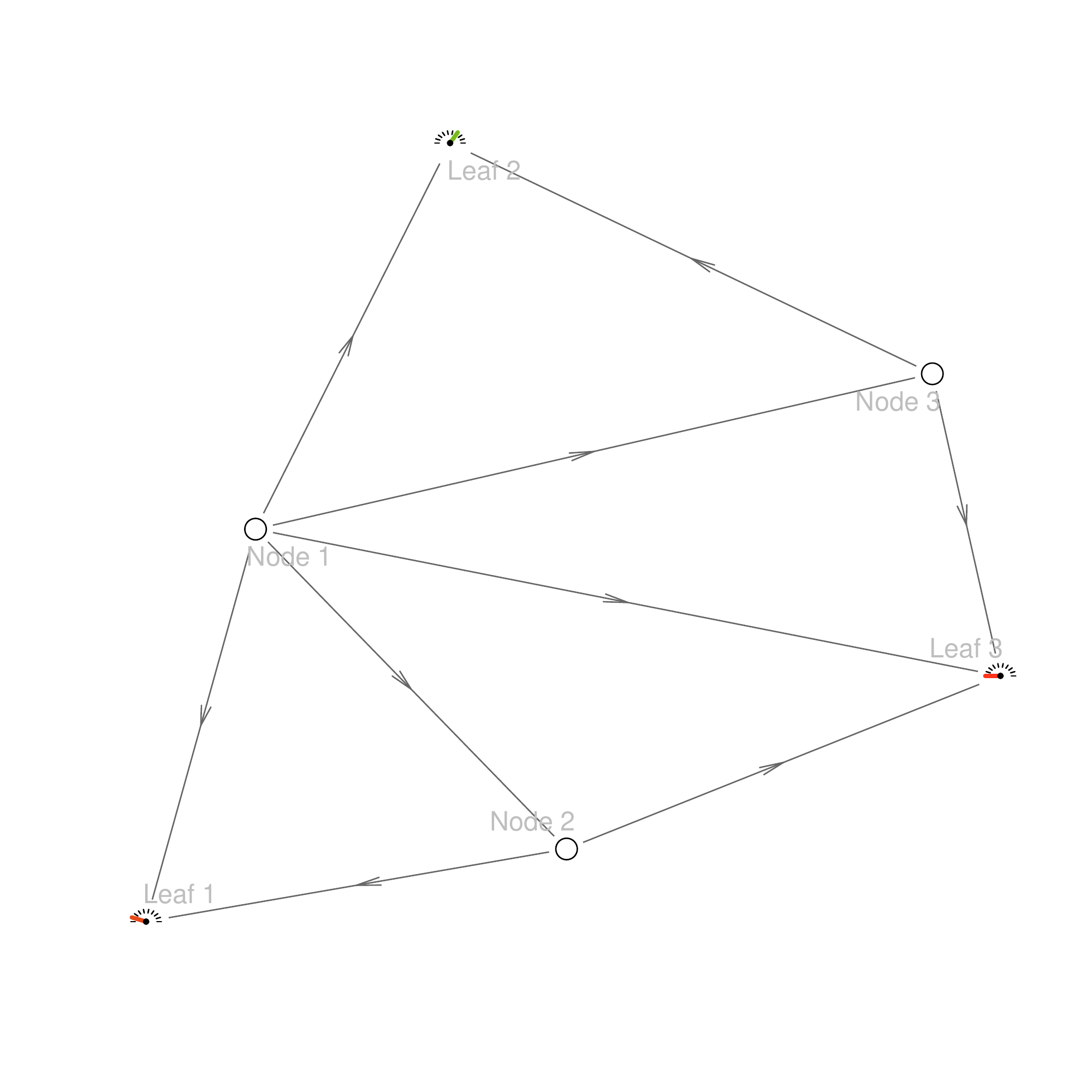}
\includegraphics[width=0.49\textwidth]{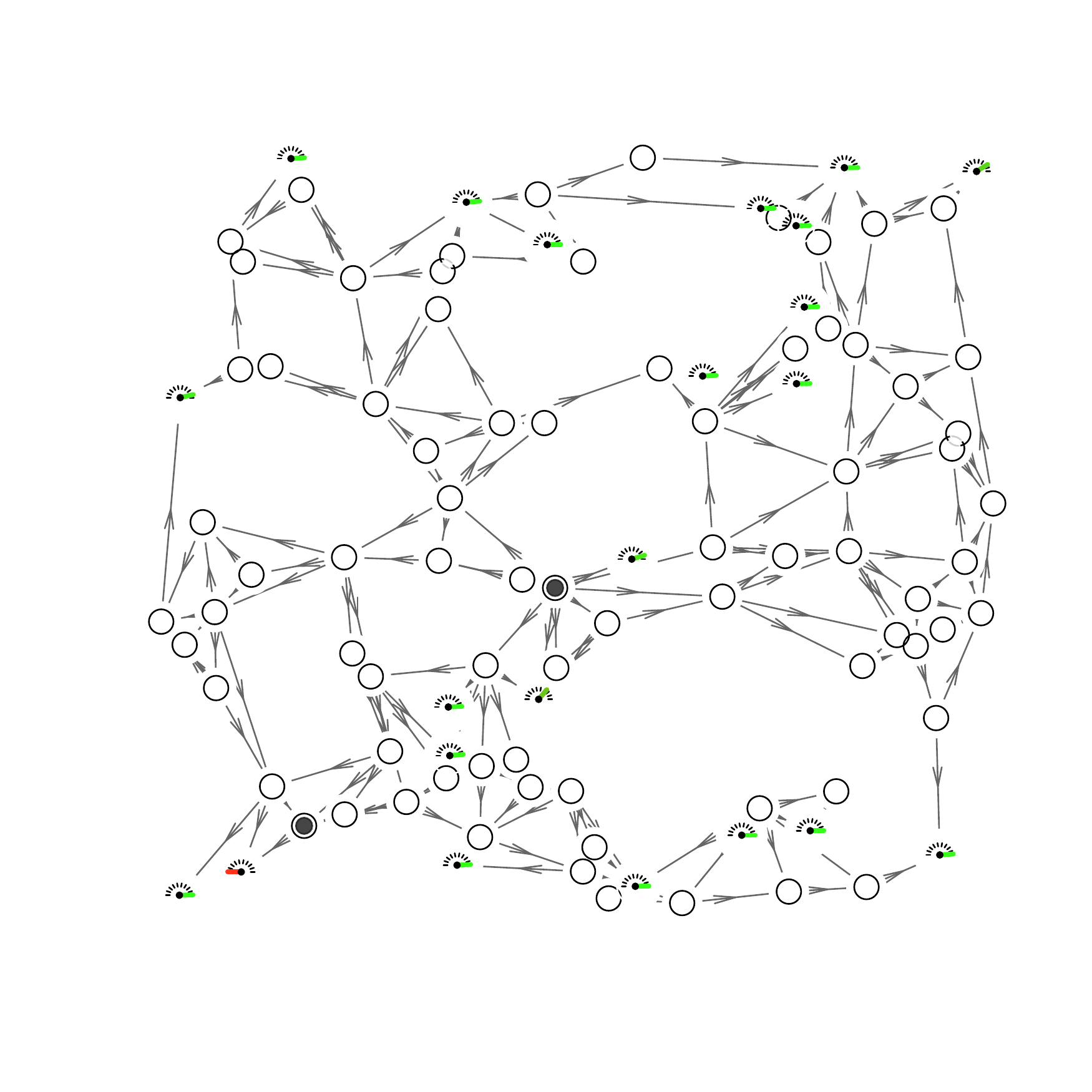}
\caption{\label{fig:1}  Left: A network with three internal nodes and three leaf nodes. The (unobservable) losses at the internal nodes are (10,10,0), meaning that the first two nodes lead to a loss rate of 10 and the third node is not leading to any losses. The observations of the loss rates at the leaf nodes are then (8,3,9). Using the observations at the leaf nodes and knowledge of the topology, NNLS can correctly identify the two first nodes as responsible for the losses. Right: A network with 78 internal nodes and 22 leaf nodes. Two of the internal nodes have a positive loss (marked with a dot) and the observations at the leaf nodes are again sufficient to pinpoint the (unknown) location of the two nodes using NNLS estimation.}
\end{center}
\end{figure}

\section{Numerical Results}
The results above imply that NNLS can be very effective if (a) the sign of regression coefficients is known or can easily be estimated and (b) the \emph{Positive Eigenvalue Condition} holds. 
\emph{Network tomography} is a good example \citep{castro2004network}. For others, including image analysis and applications in signal processing,  see \citet{slawskisparse}.  There are different aspects of network tomography, including origin-destimation matrix estimation and link-level network tomography; see \citet{castro2004network} for a good overview of the statistical aspects and 
\citet{xi2006estimating} and \citet{lawrence2006network} for a discussion of active tomography in the context of link-level analysis. We will focus on one aspect of the link-level network tomography. The network consists of nodes arranged in a directed acyclic graph (or sometimes as a special case a tree) and measurements can be taken at the leaf nodes. These measurements are used to infer the state of all the nodes in the network. 
In a communication network, the measurements can be the delay or loss rate of packages, in a transport network (such as water distributions networks) it can be the shortfall of the flow rate compared to the expected rate. Since the network topology is assumed to be known, the measurements consist typically of noise plus a linear combination of the internal and unobservable states of the nodes in the network. If a node in the network has a loss (be it in the form of delaying packages or loss of water flow), it will have a linear effect on all leaf nodes that are descendants of the node in the directed acyclic graph. 

Figure~\ref{fig:1} shows a toy example. Imagining a flow passing through the tree from the internal nodes to the leaf nodes, the entry $\mb X_{i,j}$ is the proportion of flow in node $j$ that reaches leaf node $i$ if flow is divided equally among all outgoing edges in each node of the tree.  Three internal nodes have loss rates $(\beta_1,\beta_2,\beta_3)=(10,10,0)$. The loss rates $\mb Y = (Y_1,Y_2,Y_3)$ at the three leaf nodes are then given by $\mb Y=\mb X\beta + \varepsilon$ for some i.i.d.\ noise $\varepsilon$ and 
\[ \mb X = \left( \begin{array} {ccc}    0.3 & 0.5 & 0 \\ 0.3&0&0.5 \\ 0.4&0.5&0.5 \end{array} \right) .\]
A positivity constraint on the coefficient vectors is clearly appropriate since there will in general not be a negative loss at internal nodes (for example no unexected \emph{gain} of water in a distribution network).
In the noiseless case, the NNLS solution recovers exactly the internal states $(10,10,0)$ and thus identifies correctly  the first two nodes as responsible for the loss of the flow rate in all three leaf nodes. In this simple example, the number of leaf nodes is equal to the number of internal nodes and ordinary least squares would also work in the noiseless case. Least squares clearly ceases to be useful once the number of internal nodes exceeds the number of leaf nodes.
Note that, contrary to the previous literature (for example \citet{castro2004network,lawrence2006network}) we do not attempt to  fit a stochastic model to the observations. We are merely trying to directly estimate the current internal state $\beta$ of the nodes in the network as accurately as possible. 

The theory suggests that a non-negativity constraint can already be very powerful under certain constraints on the design matrix. The main condition is the \emph{Positive Eigenvalue Condition}. In our simple network tomography example, it is obvious that all entries in $\mb X$ are positive and the same is hence true for $\hat{\S}=n^{-1} \mb X^T \mb X $. Entries in $\mb X$ correspond to the amount of loss (delay of packages or reduction in flow rate) in a leaf node caused by a specific loss at an internal node and is non-zero if and only if there is a connection between the internal and the leaf node. Suppose that all non-zero entries in $\mb X$ have entries at least as large as $\delta$ for some $\delta>0$. 
Suppose further that we can group all internal nodes into $B$ blocks such that the internal nodes within a block share at least one leaf node to which they all connect. The \emph{Positive Eigenvalue Condition} is then fulfilled with value $\delta^2/B$; see Example III in the discussion of the condition. 

The theory seems to show that under these conditions the NNLS-regularization is effective. To test this, we examine the effect of placing an additional $\ell_1$-constraint on the coefficient by computing 
 \begin{equation} \label{eq:nnlsl1}      \beta^\lambda \; :=\; \mbox{argmin}_{\beta}  \; \|\mb  Y - \mb X \beta \|^2_2 \quad \mbox{such that   } \min_k\beta_k \ge 0 \mbox{   and   } \|\beta\|_1 \le \lambda . \end{equation}
Let $\hat{\beta}$ be again the NNLS-solution defined in (\ref{eq:nnls}).  It is obvious that $\beta^\lambda \equiv \hat{\beta}$ for all $\lambda\ge \lambda_{\max}$ for $\lambda_{\max}:=\|\hat{\beta}\|_1$.

\begin{figure}
\begin{center}
\includegraphics[width=0.7\textwidth]{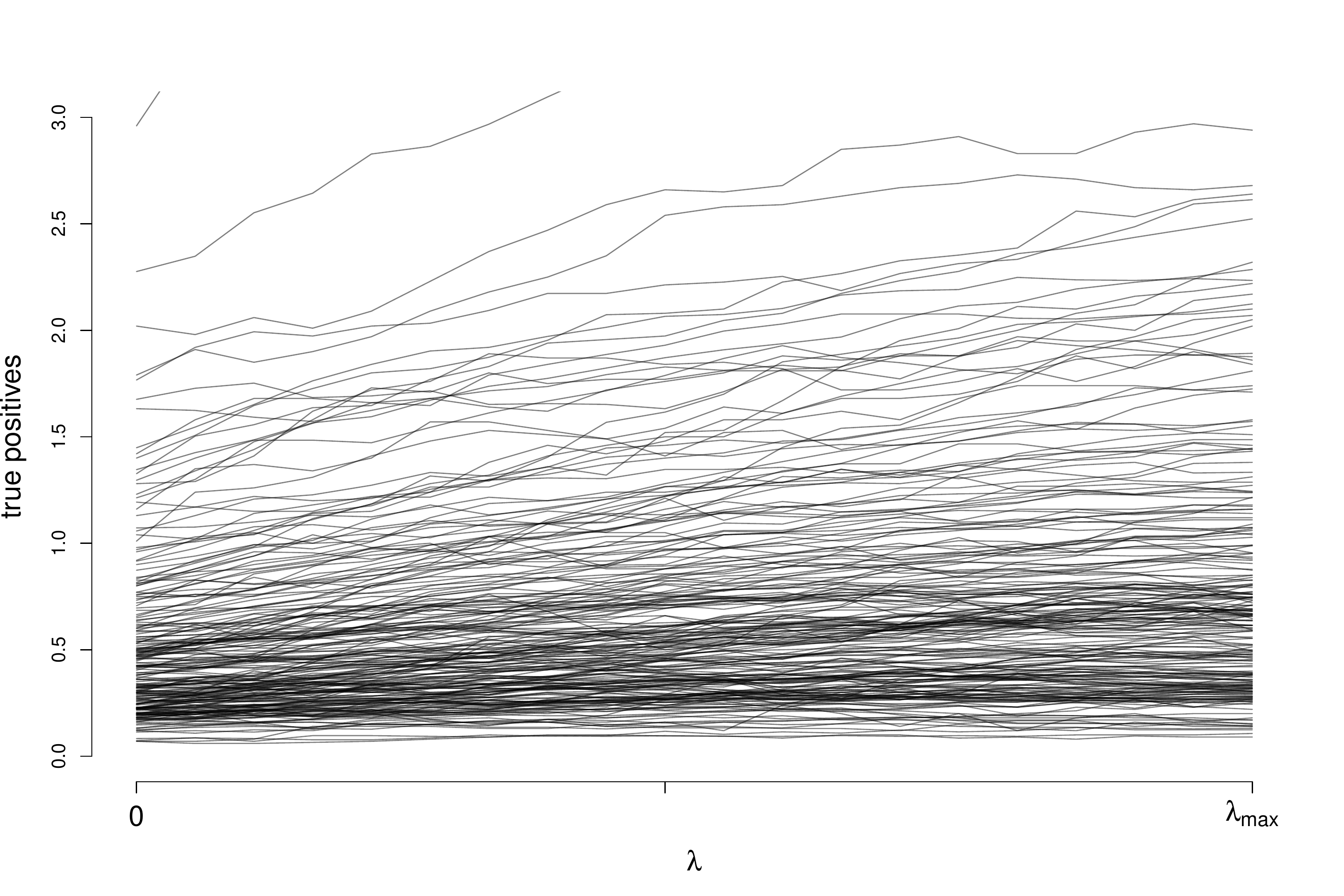}
\caption{\label{fig:lambda}  The average number of correctly identified internal nodes with  a positive loss under 1000 different scenarios with an additional $\ell_1$-constraint as in (\ref{eq:nnlsl1}). The NNLS solution corresponds to $\lambda=\lambda_{\max}$ and is seen to be in general superior to the solutions under additional shrinkage. }
\end{center}
\end{figure}

We generate networks of similar type as the ones shown in Figure~\ref{fig:1}. The number $N$ of total nodes is chosen for each of 1000 simulations uniformly out the set $\{25,50,100,200,400\}$. Nodes are distributed uniformly on the area $[-1,1]^2$
and numbered in order of their Euclidean distance from the origin. Starting with the first node $k=1$ closest to the origin, edges are drawn between it and its $K$ nearest neighbours with a larger ordering number (where $K$ is drawn uniformly from the set $\{5,10,20\}$). When drawing edges at node $k=1,\ldots,N-1$, they are deleted with probability $\nu$  (where $\nu$ is drawn uniformly from the set $\{.2,.4,.6,.8,1\}$) or when the edge would cross a previously drawn edge. Imagining again a flow passing through the tree from the internal nodes to the leaf nodes, the entry $\mb X_{i,j}$ is the proportion of flow in node $j$ that reaches leaf node $i$ if flow is divided equally among all outgoing edges in each node of the tree. For each of the 1000 simulations, we draw a single graph from the parameters as described above and also draw the noise variance uniformly from the set $\{0,0.125,0.25,0.5,1,2,4\}$ and a number $s$ of non-zero entries in $\beta$ (corresponding to nodes with a delay or loss), where $s$ is drawn uniformly from the set $\{2,5,10\}$. The $s$ non-zero entries from $\beta$ are generated independently as the absolute value of a standard-normal random variable. For each such setting, we simulate 50 times the vector $\mb Y$ and reconstruct with $\hat{\beta}^\lambda$ as in (\ref{eq:nnlsl1}) for an evenly spaced grid of 20 points between $\lambda=0$ and $\lambda=\lambda_{\max}$, the NNLS solution. Nodes are put in decreasing order of the reconstructed value $\hat{\beta}^\lambda$. We record the first entry in the re-ordered vector $\hat{\beta}^\lambda$ that corresponds to a false positive (a zero entry in the equally re-ordered  vector $\beta$) and call the number of true positives the number of values of $\hat{\beta}^\lambda$ with larger value than the first false positive.

Figure \ref{fig:lambda} shows the average number of true positives as a function of $\lambda$. Each line corresponds to the average value over all 50 simulations in a given scenario. For nearly all scenarios there is no benefit in placing an additional $\ell_1$-penalty on the coefficients. The NNLS solution is thus a very good and simple estimator in these settings, as expected from theory.  Additional regularization by an $\ell_1$-penalty does not seem to improve results.

\section{Discussion}
We have shown that non-negative (or sign-constrained) least squares can be an effective regularization technique for sparse high-dimensional data  under two conditions: (a) the data fulfil the so-called \emph{Positive Eigenvalue Condition}, which is easy to check for a given dataset, and (b) the sign of the coefficients is known or can easily be estimated. 
If the conditions hold, NNLS can recover the correct sparsity pattern in the absence of any further shrinkage, as long as $\log(p) s/n\rightarrow 0$ for $n\rightarrow\infty$, where $p$ is the number of variables, $s$ the number of non-zero variables in the optimal regression vector and $n$ is sample size. 
We have shown network tomography as an example where the sign of regression coefficients is known a priori and the design condition is fulfilled automatically, at least approximately. In other examples the sign can be estimated by an initial estimator. An attractive feature of NNLS is that it does not require any tuning parameter beyond the choice of the signs of the individual regression coefficients. Despite its simplicity, it can remarkably accurate for high-dimensional regression.

\section{Appendix: Proofs}

\subsection{Proof of Theorem~\ref{theorem:main}}

First, for any $C>0$, $1-\Phi(C) \le (2\pi)^{-1/2} C^{-1} \exp(-C^2/2)$. Choosing $C^2=K^2_{p,\eta} =2 \log(\frac{\sqrt 2 p}{\sqrt{\pi} \eta})$, it follows with $\eta<1/3$ and hence $C \ge 1$  that $1-\Phi(C)  \le \eta/(2p) $.  Thus $1-(p+s)(1-\Phi(C))\ge 1-(2p) (1-\Phi(C)) \ge 1-\eta$ and the results follow hence from Lemma~\ref{lemma:main}.

\subsection{Proof of Theorem~\ref{theorem:pred}}
Define the oracle non-negative least squares solution as 
\begin{equation} \label{eq:nnlsopt}      \hat{\beta}^{oracle} \; :=\; \mbox{argmin}_{\beta}  \| \mb Y -\mb X \beta\|_2^2  \quad \mbox{such that} \;\; \min_k\beta_k \ge 0  \;\;\mbox{and}\;\; \beta_N \equiv 0 , \end{equation}
and let $\delta \beta =  \hat{\beta} - \hat{\beta}^{oracle} $. 

Let $M$ be the set $M \; :=\; \{k: \delta\beta_k <0\} $.
Using Equation (\ref{eq:h1}) in the proof of Lemma~\ref{lemma:main},  it follows that, with probability at least $1-(p+s)(1-\Phi(C))$,
\[  \delta\beta^T \hat{\S} \delta \beta \;\le\;  2 C \sigma \|\delta \beta_{M^c}\|_1/\sqrt{n}  \]
and, using $\| \delta\beta_{M^c}\|_1\le\|\delta\beta\|_1$ and the bound in (\ref{eq:toshowlem}) for the latter quantity, it holds with probability at least $1-(p+s)(1-\Phi(C))$ that 
\[ \delta\beta^T \hat{\S} \delta \beta \;\le\;  2 C^2 \sigma^2 (5\nu^{-1} +2\phi^{-1/2}) \frac s n.\]    
Using again $C^2 = K^2_{p,\eta}= 2 \log(\frac{\sqrt 2 p}{\sqrt{\pi} \eta})$, the claim follows.
\subsection{Lemmata}

\begin{lemma}\label{lemma:main}
Assume that the \emph{Positive Eigenvalue Condition} holds with $\nu>0$. Choose any $C>0$.
Assume that the compatibility condition holds with $\phi>0$  for $L=4\nu^{-1}$ and $\min_{k\in S }\beta_k > C \sigma/\sqrt{n \phi}$. It then holds with probability at least $1-(p+s)(1-\Phi(C))$ that  
\[ \|\hat{\beta}-\beta^*\|_1 \; \le \;  C\sigma (5/\nu + 4/\sqrt{\phi})   \frac{s }{\sqrt{n}}  .\] 
\end{lemma}

\begin{proof}By the definition (\ref{eq:nnls}) of $\hat{\beta}$ and definition (\ref{eq:nnlsopt}) of $\hat{\beta}^{oracle}$,
\begin{equation}\label{eq:deltabetaopt}    \delta \beta \;=\; \mbox{argmin}_\gamma \; \|\mb Y - \mb X \hat{\beta}^{oracle}  -\mb X \gamma\|_2^2 \;\;\mbox{   such that   } \;\;  \gamma _k  \ge -\hat{\beta}^{oracle}_k \; \mbox{for all } k=1,\ldots,p  . \end{equation}
The bound for $\|\hat{\beta}-\beta^*\|_1$ follows as 
$ \|\hat{\beta}-\beta^*\|_1 \le \|\hat{\beta}^{oracle}- \beta^*\|_1 +  \|\delta \beta\|_1 .$
Using Lemma \ref{lemma:equiv}, it holds with probability exceeding $1-(p+s)(1-\Phi(C))$,
\begin{equation}\label{eq:l3} \|\hat{\beta}^{oracle}- \beta^*\|_1 \le 2C \sigma \phi^{-1/2} \frac{s}{ \sqrt{ n}} ,\end{equation}
and it thus remains to be shown that, if (\ref{eq:l3}) is fulfilled, also
\begin{equation}\label{eq:toshowlem} \|\delta \beta \|_1 \le  C\sigma (5\nu^{-1} +2 \phi^{-1/2}) \frac{s}{\sqrt{n}} .\end{equation}

Let $\mb R= \mb Y - \mb X \hat{\beta}^{oracle}$.  Since $\delta \beta\equiv 0$ is a feasible solution in (\ref{eq:deltabetaopt}),  we have that 
\[  \delta\beta^T \mb X^T \mb X \delta \beta - 2 \mb R^T \mb X \delta \beta \le 0.\]

Let 
\begin{equation}\label{eq:M}
M \; :=\; \{k: \delta\beta_k <0\} \end{equation}
By the definition of the estimator $M\subseteq S$ and $N \subseteq M^c$.
By Lemma \ref{lemma:grad}, with probability at least $1-p(1-\Phi(C))$,
\[ \max_{k\in N}\; \mb R^T \mb X_k \;\le\; C \sigma \sqrt{n}  .\]
By Lemma \ref{lemma:equiv}, with probability at least $1-s(1-\Phi(C))$,
$ \mb R^T \mb X_k =0$ for all $k\in S$.
Hence, taken together, with probability at least $1-(p+s)(1-\Phi(C))$, 
\[   \mb R^T \mb X \delta \beta \le \big(\max_{k\in M^c}\; \mb R^T \mb X_k\big) \| \delta \beta_{M^c}\|_1 \;\le\; C \sigma \sqrt{n}  \|\delta \beta_{M^c}\|_1   \]
and thus
\begin{equation} \label{eq:h1}  \delta\beta^T \hat{\S} \delta \beta \;\le\;  2 C \sigma \|\delta \beta_{M^c}\|_1/\sqrt{n} .\end{equation}
Now, \begin{eqnarray}  \delta\beta^T \hat{\S} \delta \beta &=& \delta\beta_M^T \hat{\S} \delta \beta_M + \delta\beta_{M^c}^T \hat{\S} \delta \beta_{M^c}  - 2\sum_{i\in M, j\in M^c} \hat{\S}_{i,j} \delta \beta_i \delta \beta_j \nonumber \\ &\ge &  \delta\beta_{M^c}^T \hat{\S} \delta \beta_{M^c}  -2 \|\delta\beta_M\|_1 \|\delta \beta_{M^c}\|_1 \nonumber \\ &\ge & \nu \|\delta\beta_{M^c}\|^2_1 -2 \|\delta\beta_M\|_1\|\delta\beta_{M^c}\|_1 ,\end{eqnarray}
having used the normalization to 1 of all columns of $\mb X$ (which bounds the absolute values of all entries in $\hat{\S}$ by 1) for the second term in the second last inequality and the \emph{Positive Eigenvalue Condition} for the first term in the last inequality (together with the fact that $\min_{k\in M^c} \beta_k \ge 0$ by definition of $M$ in (\ref{eq:M})). Using this bound in (\ref{eq:h1}) and dividing by $\|\delta \beta_{M^c}\|_1$ yields that, with probability at least $1-(p+s)(1-\Phi(C))$, 
\begin{eqnarray} \|\delta\beta_{M^c}\|_1 &\le& 2 \nu^{-1}  \Big(  C \sigma / \sqrt{n} + \|\delta\beta_M\|_1\Big)\nonumber \\ & =& 2 \nu^{-1}  \Big(  \frac{C \sigma }{\|\delta\beta_M\|_1 \sqrt{n}} +1 \Big)   \|\delta\beta_M\|_1 \label{eq:h2}  \end{eqnarray}
Evidently  $\|\delta \beta_M\|_1\le C\sigma/\sqrt{n}$ is either true or not. If it is true, then it follows trivially from the first inequality in  (\ref{eq:h2}) that $\|\delta\beta_{M^c}\|_1\le 4\nu^{-1} C\sigma/\sqrt{n}$ and hence $\|\delta\beta\|_1 \le (1+4\nu^{-1}) C\sigma/\sqrt{n} \le 5\nu^{-1} C\sigma /\sqrt{n}$, and the bound in (\ref{eq:toshowlem}) holds true. 

Alternatively, if $\|\delta \beta_M \|_1 > C\sigma /\sqrt{n}$ we have from the second inequality in (\ref{eq:h2}) that $\|\delta\beta_{M^c}\|_1\le L \|\delta\beta_M\|_1$ for $L=4\nu^{-1}$ and thus, using $N\subseteq M^c$, also $\|\delta\beta_{N}\|_1\le L \|\delta\beta_S\|_1$. The vector $\delta \beta$ is then in $\mathcal{R}(L,S)$.
Using  the compatibility condition, it follows that  \[ \delta\beta^T \hat{\S} \delta\beta \; \ge \; \frac{\phi}{s} \| \delta \beta \|_1^2  .\] Using this in (\ref{eq:h1}), 
\begin{equation}\label{eq:hhm} \frac{\phi}{s} \|\delta \beta\|_1^2 \; \le\;  2C\sigma \|\delta \beta_{M^c}\|_1/\sqrt{n} .\end{equation}
Using $\|\delta\beta_{M^c}\|_1\le \|\delta\beta\|_1 $, it follows that  \begin{equation}\label{eq:end} \|\delta\beta\|_1 \le 2C \sigma s/\sqrt{ n \phi} ,\end{equation} which also satisfies the bound in (\ref{eq:toshowlem}). Hence, the bound (\ref{eq:toshowlem}) holds under both possible scenarios ($\|\delta \beta_M\|_1\le C\sigma/\sqrt{n}$  true or false) and 
the proof is complete.

 \end{proof}

\begin{lemma}\label{lemma:equiv}
Let $\hat{\beta}^{ols}$ be the least squares estimator restricted to $S$:
\[ \hat{\beta}^{ols} \; =\;   \mbox{argmin}_\beta \; \| \mb Y - \mb X\beta\|_2^2 \;\;\mbox{such that}\;\; \beta_N \equiv 0  .   \]
If $\min_{k\in S} \beta^*_k \ge C\sigma /\sqrt{n \phi}$,
Then 
\[ P( \hat{\beta}^{ols} \equiv \hat{\beta}^{oracle}) \; \ge \; 1- s(1-  \Phi(C)) ,\] 
and, with at least the same probability $1-s(1-\Phi(C))$,
\[ \| \beta^* - \hat{\beta}^{oracle} \|_\infty \; \le \;  C \sigma / \sqrt{n \phi} \]
\end{lemma}
\begin{proof}
It is only necessary to show that $\min_{ k \in S} \hat{\beta}^{ols}_k \ge 0$ with probability at least $1 - s(1-\Phi(C))$. 

The error term has, under the made assumptions, a normal distribution, $\hat{\beta}^{ols}_k -\beta^* \sim \mathcal{N}(0, \sigma^2 (n \hat{\S}_{SS})^{-1}_k )$ for all $k\in S$. The minimal eigenvalue of $\hat{\S}_{SS}$ is bounded from below by $\phi$ by the compatibility condition and the variance of $\hat{\beta}^{ols}_k$ is thus bounded from above by $\phi^{-1}\sigma^2/n$ for all $k\in S$. 
 It follows with Bonferroni's inequality that, with probability at least $1-s(1-\Phi(C))$, \begin{equation}\label{eq:linfols} \|\beta^* - \hat{\beta}^{ols} \|_\infty \le C \sigma /\sqrt{n\phi}.\end{equation}
If $\min_{k\in S} \beta^*_k \ge C\sigma /\sqrt{n \phi}$, then (\ref{eq:linfols}) implies that $\min_{k\in S} \hat{\beta}^{ols}_k \ge 0$ and thus $\hat{\beta}^{oracle} \equiv \hat{\beta}^{ols}$ and thus also \[ \|\beta^* - \hat{\beta}^{oracle} \|_\infty \le C \sigma /\sqrt{n\phi} ,\]
which completes the proof. 
\end{proof}

\begin{lemma}\label{lemma:grad}
With probability at least $1- p (1-\Phi(C))$, 
\[  \max_{k \in N} (\mb Y - \mb X \hat{\beta}^{oracle})^T \mb X_k  \le C\sigma \sqrt{n} . \]
\end{lemma}
\begin{proof}
We condition on the event $\hat{\beta}^{oracle} \equiv \hat{\beta}^{ols}$, which happens according to Lemma \ref{lemma:equiv} with probability at least $1-s(1-\Phi(C))$. Then $\mb Y  - \mb X \hat{\beta}^{oracle} = \mb Y - \mb X \hat{\beta}^{ols} =P_{S\perp} \mb Y$, where $P_{S\perp} \mb Z$ is the projection of a vector $\mb Z\in\mathbb{R}^n$ into the space orthogonal to $\mb X_S$. Now , $P_{S\perp} \mb Y = P_{S\perp} ( \mb X \beta^* + \varepsilon) = P_{S\perp} \varepsilon$.   The distribution of $(P_{S\perp} \varepsilon)^T \mb X_k$ is, for every $k\in N$, normal with mean 0 and variance at most $\sigma^2n$, and thus $P( (P_{S\perp} \varepsilon)^T \mb X_k \ge C \sigma \sqrt{n}) \le 1-\Phi(C)$ for all $k\in N$ and, using a Bonferroni bound, $P( \max_{k\in N} (P_{S\perp} \varepsilon)^T \mb X_k \ge C \sigma \sqrt{n}) \le |N| (1-\Phi(C))$. The unconditional probability of $\max_{k\in N} (P_{S\perp} \varepsilon)^T \mb X_k \ge C \sigma \sqrt{n}$ is thus at least $1-s(1-\Phi(C))-|N|(1-\Phi(C)) = 1-(s+|N|) (1-\Phi(C)) = 1-p(1-\Phi(C))$, which completes the proof.

\end{proof}

\end{document}